\documentclass[11pt,a4paper,reqno]{amsart}
\pdfoutput=1

\usepackage{amsthm,amsmath,mathtools,siunitx,fancybox,xcolor,caption,subcaption,graphicx}
\usepackage[percent]{overpic}
\usepackage[utf8]{inputenc}

\newtheorem{theorem}{Theorem}

\newtheorem{lemma}[theorem]{Lemma}
\newtheorem{definition}[theorem]{Definition}
\newtheorem{remark}[theorem]{Remark}

\usepackage{hyperref}

\title[Conformal neutrino tomography]{Breaking the conformal freedom of spacetime with supernova neutrino imaging}
\author{Joonas Ilmavirta and Gunther Uhlmann}
\thanks{J. Ilmavirta, Department of Mathematics and Statistics, University of Jyv\"askyl\"a, P.O. Box 35 (MaD), FI-40014 University of Jyväskylä, Finland. \texttt{joonas.ilmavirta@jyu.fi}}
\thanks{G. Uhlmann, Department of Mathematics, University of Washington, Seattle, WA 98195-4350, USA, and Institute for Advanced Study of the Hong Kong University of Science and Technology. \texttt{gunther@math.washington.edu}}
\date{\today}

\newcommand{\R}{\mathbb R}

\newcommand{\der}{\mathrm d}

\newcommand{\ip}[2]{\left\langle#1,#2\right\rangle}

\begin{document}

\begin{abstract}
It is known that a geometric measurement of the light cones of supernovae determines the conformal class of the visible part of the spacetime.
The conformal factor is physically meaningful but cannot be determined geometrically by anything with zero mass, such as the photon.
We show that measuring the neutrino cones in addition to light cones completely removes this gauge freedom.
We describe the physical model in great detail, including why ultrarelativistic neutrinos are the only option.
\end{abstract}

\maketitle

\section{Introduction}

Can a color-blind astronomer reconstruct a reliable model of the universe by just looking at the sky?
More accurately:
If one makes geometric measurements of the arrivals of all photons from all supernova explosions but has no spectral information on the photons, can one reconstruct the Lorentzian metric describing the spacetime in one's visible past?

This is not quite possible:
Such data does determine the conformal class of the metric, but nothing more can be said from such data.
Lightlike geodesics are conformally invariant.

This conformal gauge freedom can be broken by measuring a particle that has a non-zero mass, and the only such option is a neutrino.
We will show that in a simple but physically reasonable model neutrino measurements fix the conformal factor uniquely.

Our data is purely geometric in nature.
We only assume that we know the mechanism behind the supernovae (which is valid for type Ia); all measurements of photons and neutrinos are based on the light cones and ``neutrino cones'' that describe when and where they are observed.
All measurements are passive, as otherwise it is impossible to reach cosmologically relevant scales in any remotely realistic scenario.

We model the trajectories of neutrinos as first order perturbations of those of photons.
We will see later that this linearization is well justified, as the main small parameter~$\mu^2$ describing neutrino mass is on the order of~$10^{-20}$.

We begin the article with the geometric background in section~\ref{sec:math-intro}, covering the results in inverse problems related to light observation sets and determination of the conformal class.
We then give a detailed exposition of our geometric model for supernova neutrinos in section~\ref{sec:model}.
The geometric and the physical introductions can be read independently of each other.
We state and prove our main theorem in section~\ref{sec:thm}.

\subsection{Geometric background}
\label{sec:math-intro}

Let $(M, g)$ be a $(1+3)$-dimensional time oriented Lorentzian manifold.
The signature of~$g$ is $(+, -, -, -)$.
(The typical choice of signature is $(-, +, +, +)$ in the inverse problems literature, but our choice is more natural for particle physics so we use it to simplify kinematic considerations below.)
The prototypical example is  Minkowski space-time $(\mathbb{R}^4, g_m)$ with the metric $g_m = \der t^2 - \der x^2 - \der y^2 - \der z^2$.
The light cone and the classification of tangent vectors is illustrated in figure~\ref{fig:cone}.

\begin{figure}[ht]
	\centering
	\begin{overpic}[scale=.3]{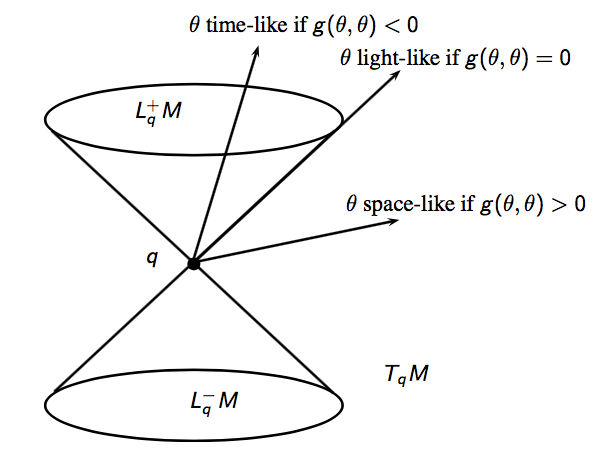}
    \put (-15,70) {\small\colorbox{white}{$v$ is timelike if $g(v,v)>0\qquad\qquad$}}
    \put (55,65.5) {\small\colorbox{white}{$v$ is lightlike if $g(v,v)=0$}}
    \put (55,42) {\small\colorbox{white}{$v$ is spacelike if $g(v,v)<0$}}
    \end{overpic}
    \caption{The light cone at $q\in M$ is a subset of~$T_qM$. In our sign conventions timelike directions have a positive square.}
    \label{fig:cone}
\end{figure}

\begin{definition}
a)	$L_q^\pm M$ is the set of future (past) pointing lightlike vectors at~$q$.

b)  Casual vectors are the collection of timelike and lightlike vectors.

c)	A curve~$\gamma$ is timelike (lightlike, causal) if the tangent vectors are timelike (lightlike, causal).
\end{definition}

Let~$\hat \mu$ be a timelike geodesic, which corresponds to the worldline of an observer in general relativity. 
For $p, q\in M$ the notation ${p\ll q}$ means that~$p$ and~$q$ can be joined by a future pointing timelike curve, and ${p<q}$ means that~$p$ and~$q$ can be joined by a future-pointing causal curve. 

\begin{figure}[ht]
	\centering
	{\includegraphics[scale=.3]{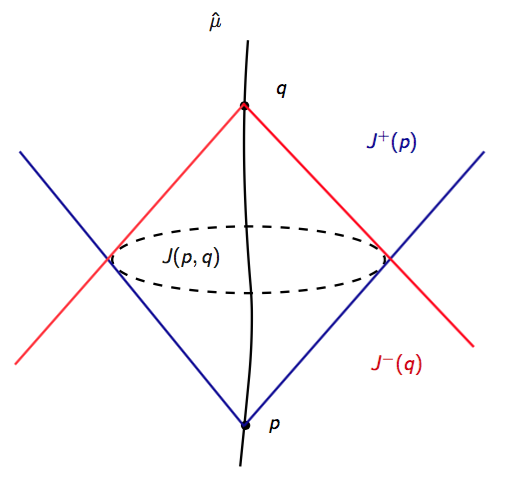}}
	\caption{The intersection of the causal future~$J^+(p)$ of~$p$ and the causal past~$J^-(q)$ of~$q$ is the diamond-shaped region $J(p,q)\subset M$.}
	\label{fig:diamond}
\end{figure}

\begin{definition}
a)	The chronological future of $p\in M$ is $I^+(p) = \{q\in M: p\ll q\}$.

b)	The {causal future} of $p\in M$ is $ J^+(p) = \{q\in M: q< p\}$.

c)   $J(p, q) = J^+(p)\cap J^-(q)$, $I(p, q) = I^+(p)\cap I^-(q)$ are the diamond-shaped regions as depicted in figure~\ref{fig:diamond}.
\end{definition}

\begin{definition}
A Lorentzian manifold $(M, g)$ is {globally hyperbolic} if
there is no closed causal paths in~$M$, and for any $p, q\in M$ and $p<q$, the set $J(p, q)$ is compact.
\end{definition}

Under the global hyperbolicity assumption,  Einstein's equations, for instance,  are well-posed on $(M, g)$.
Also, in this case,  $(M, g)$ is isometric to the product manifold 
$
{\R\times N}
$
with the metric 
$
{g =\beta(t, y)\der t^2 - \kappa(t, y)}
$
.
Here $\beta: \R\times N\to \R_+$ is smooth, $N$ is a three-dimensional manifold and~$\kappa$ is a Riemannian metric on~$N$ and smooth in~$t$.  

We shall use $x = (t, y) = (x_0, x_1, x_2, x_3)$ as the local coordinates on~$M$. 
Let $\mu=\mu([-1,1])\subset M$ be a timelike geodesic containing~$p^-$ and~$p^+$.
We consider observations in a neighborhood $U\subset M$ of~$\mu$.

\begin{definition} 
Let $W\subset I^-(p^+)\setminus J^-(p^-)$ be relatively compact  open set.

The \emph{light observation set} for $q\in W$ is
\begin{equation}
{P_U(q)\coloneqq\{\gamma_{q,\xi}(r)\in U;\ r\geq 0,\ \xi\in L_q^{+}M\}.}
\end{equation}
See figure~\ref{fig:two}.
\end{definition}

\begin{figure}
\centering
\begin{subfigure}{.5\textwidth}
  \centering
  \begin{overpic}[height=4cm]{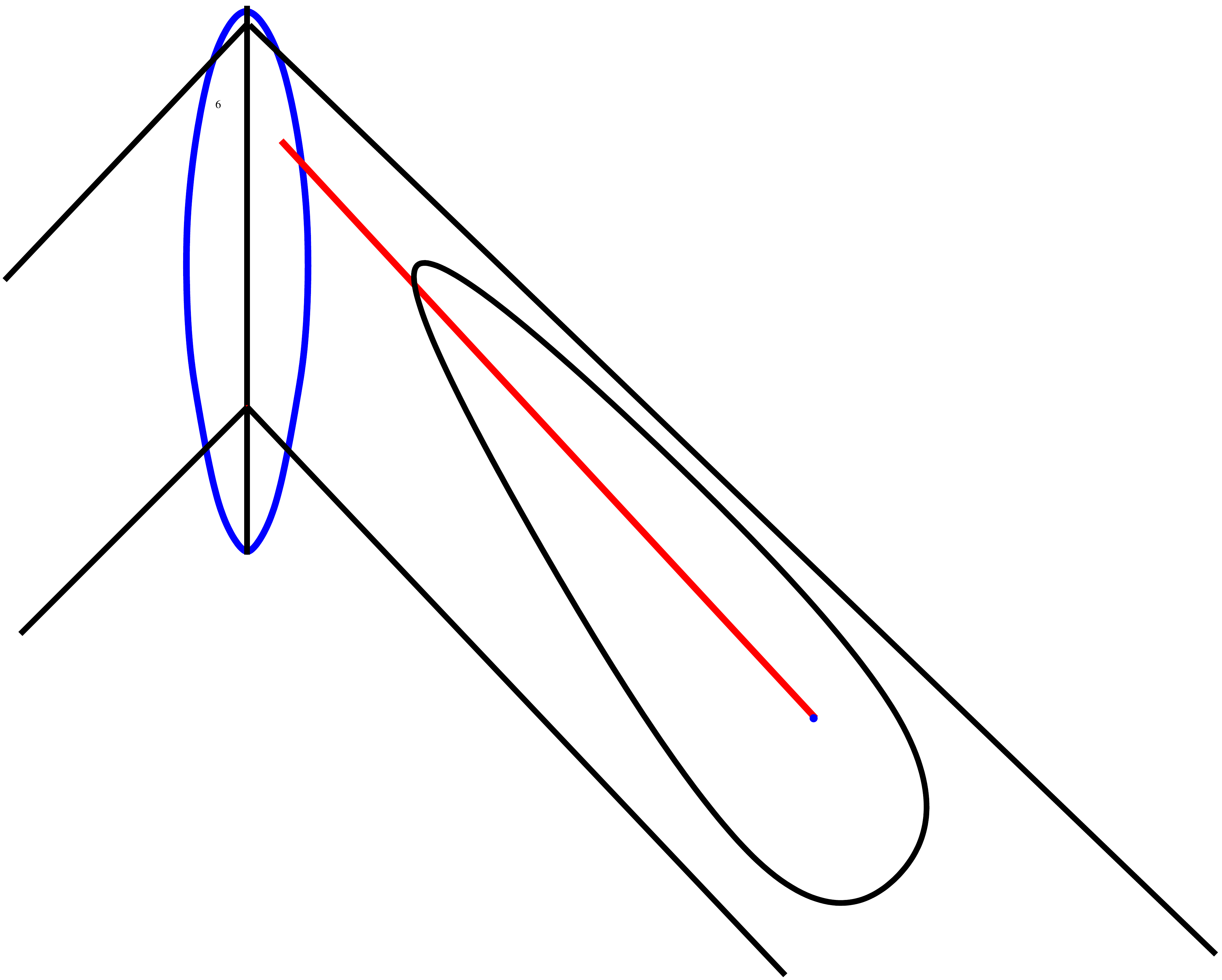}
    \put (22,82) {\small$p^+$}
    \put (22,30) {\small$p^-$}
    \put (10,50) {\small$U$}
    \put (62,12) {\small$W$}
    \end{overpic}
\end{subfigure}%
\begin{subfigure}{.5\textwidth}
  \centering
  \begin{overpic}[height=4cm]{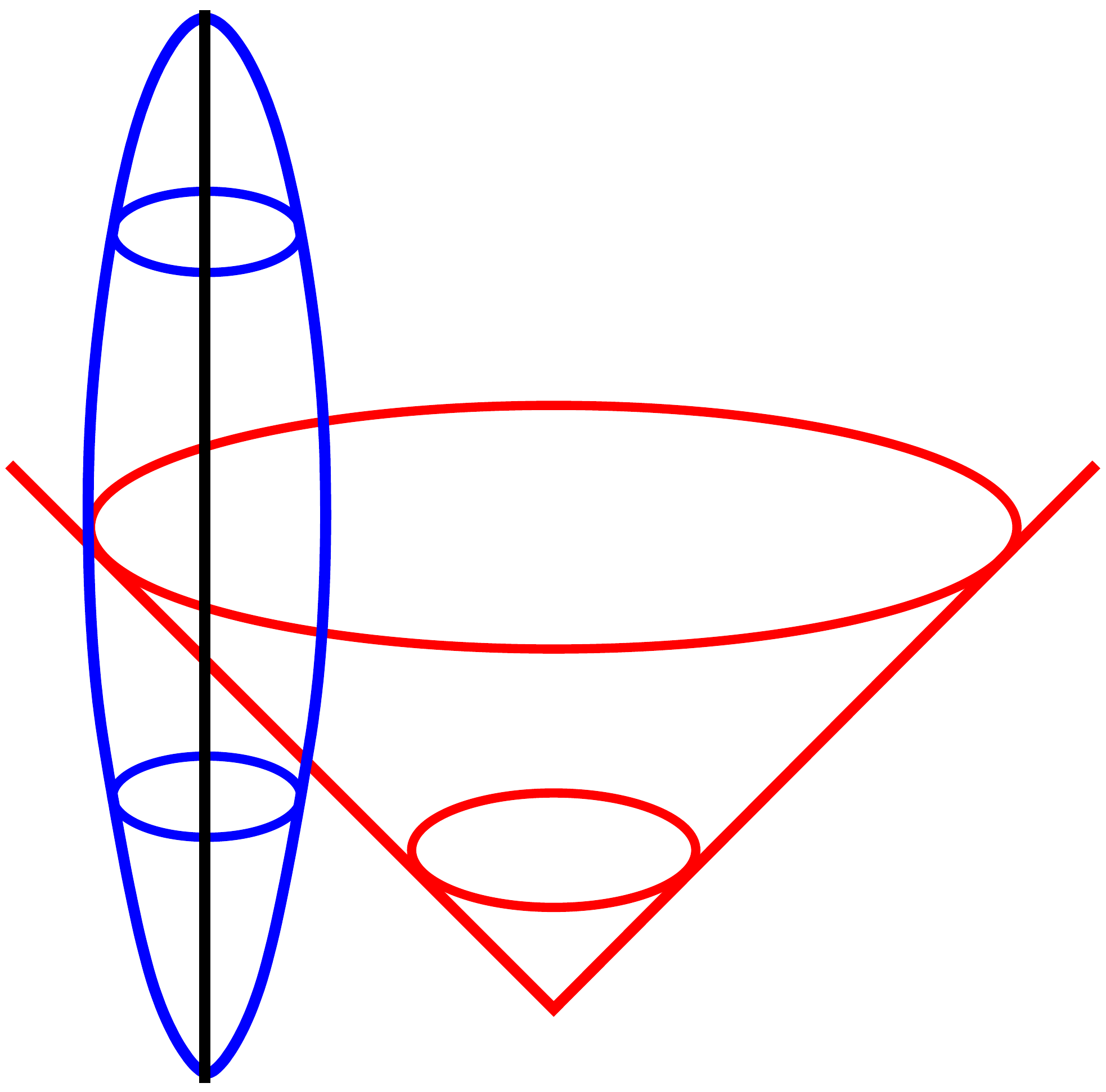}
  \put(55,5){$q$}
  \put(25,50){\vector(1,1){25}}
  \put(50,78){$P_{U}(q)$}
    \end{overpic}
\end{subfigure}
\caption{
Left
: The blue observation set~$U$ as a neighborhood of a timelike geodesic from~$p^-$ to~$p^+$, and a set~$W$ contained in the lightlike past of~$U$.
Right
: The light observation set~$P_U(q)$ of the source point~$q$ as observed in the (blue) set~$U$.
}
\label{fig:two}
\end{figure}
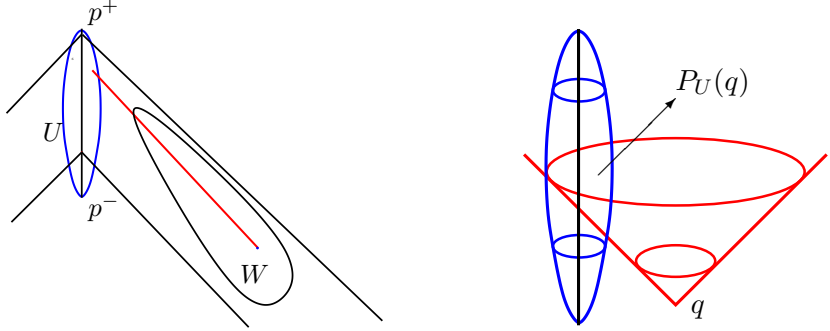

\begin{definition}
The earliest light observation set of $q\in M$ in~$U$ is 
\begin{equation}
\begin{gathered}
\mathcal{E}_U(q) = \{x\in \mathcal{P}_U(q): \text{ there is no $y\in \mathcal{P}_U(q)$ and future pointing} \\
\text{ timelike path $\alpha$ such that $\alpha(0) = y$ and $\alpha(1) = x$}\} \subset U.
\end{gathered}
\end{equation}
\end{definition}

In the physics literature the light observation sets are called light-cone cuts \cite{EH1,EH2}.

The following result was proven in \cite{KLU2013} (the proof was published in \cite{KLU2018}).

\begin{theorem}[{\cite{KLU2013,KLU2018}}]
\label{thm:klu}
	Let $(M, g)$ be an open smooth globally hyperbolic Lorentzian manifold of dimension $n\geq 3$ and let $p^+, p^-\in M$ be the points of a timelike geodesic $\hat\mu([-1, 1])\subset M, p^\pm = \hat \mu(s_\pm)$. Let $U\subset M$ be a neighborhood of $\hat\mu([-1, 1])$ and $W\subset M$ be a relatively compact set. Assume that we know 
	\begin{equation}
	 \mathcal{E}_{U}(W).
	\end{equation}
	Then we can determine the topological structure, the differential structure, and the conformal structure of~$W$, up to diffeomorphism.
\end{theorem}

Our result, theorem~\ref{thm:main}, builds on theorem~\ref{thm:klu}, using additional data to determine the only remaining unknown: the conformal factor.

The light observation sets corresponds to measurements of light point sources. This concept has also been applied to {\sl active measurements}. In this cases different types of waves are sent from a neighborhood~$U$  of a timelike geodesic. One creates an artificial source on the diamond set $I(p^-, p^+)$ that comes back to the set~$U$ and thus measuring the light observation set. This was done in \cite{KLU2013} and \cite{KLUO} for Einstein equations coupled with scalar fields, for semilinear equations in \cite{LUW1} for equations with quadratic non-linearities in the derivatives \cite{WZ}. The measurements made in those papers are encoded in the source to solution map. For the case of measurements on timelike boundary of Lorentzian manifolds with boundary in \cite{HU} was introduced the concept of {\sl boundary light observation set}. It was used to solve inverse boundary value problems by determining the Lorentzian manifold up to a conformal class in \cite{UZ} and \cite{HUZ2} by measuring the Dirichlet-to-Neumann or Neumann to Dirichlet map. For  comprehensive reviews of a very active field see \cite{L} and \cite{UZ-review}.

Linearizations of geometric and analytic inverse problems often lead to problems in integral geometry.
In the case of Lorentzian geometry and hyperbolic equations, the linear inverse problems typically concern the light ray transform.
In light light ray tomography of tensor fields, there is a conformal kernel in addition to the usual potential kernel of Riemannian ray transforms~\cite{FIO}.
For more on inverse problems in integral, we refer the reader to the review~\cite{IM}.

\section{The physical model}
\label{sec:model}

The physical basics needed to set up the model can be found in a number of books and other sources.
The most comprehensive material for neutrinos and supernovae are ~\cite{Giunti-book,SN-book}.
A reader in need of more information on neutrino physics and supernovae should be amply satisfied with this pair of books.
We will not include citations in the physical discussion of this section, and we have tried to make this section readable to a mathematician.

\subsection{Breaking symmetries with mass}

Massless particles experience no time and are blind to distance scales in the spacetime; they only care about the conformal class of the geometry.
Mass brings scale to the universe and breaks the conformal freedom that a universe consisting of only massless particles would enjoy.
Therefore massive particles are needed fix the scale everywhere, and this amounts to fixing a metric from a conformal class.

Travelling at the speed of light is a has a clear invariant description in Lorentzian geometry --- the curve is null --- and this description enjoys conformal invariance as well.
Travelling at almost the speed of light is a far uglier concept: coordinate invariance and conformal invariance are both lost.

This loss of invariance both a curse and a blessing.
Any particle travelling at less than the speed of light has a valid rest frame and in that frame it is nowhere near the speed of light.
The geometric description is far more involved and requires the use of certain natural coordinate systems.
This is the price to pay for the benefit:
breaking the conformal symmetry is precisely what allows us to identify the correct metric from a conformal class.

Particles travelling at speeds comparable to the speed of light are called relativistic, and those that are extremely close to that speed limit are called \emph{ultrarelativistic}.
Both of these concepts depend on the observer --- the only speed on which everyone will agree is exactly the speed of light.

Gravitational waves travel at the speed of light.
Geometric measurements of the kind we are using make no difference between a photon and a graviton, and therefore measuring ``graviton cones'' would not allow breaking the conformal gauge symmetry.

\subsection{Supernova neutrinos}

In order to measure the conformal factor geometrically at cosmological distances, one needs to measure particles that
\begin{enumerate}
\item
are slow enough and fast enough
(not exactly the speed of light but not a tiny fraction either),
\item
do not interact too much with electromagnetic fields near us
(so that they are accurately modelled by geodesics),
\item
can be detected,
and
\item
are produced in sufficient quantities all around the universe.
\end{enumerate}
The only particle that meets these requirements is the neutrino:
Neutrinos have tiny but non-zero mass and are typically ultrarelativistic.
(The slower ones are practically impossible to observe with current technology.)
They are electrically neutral and interact very weakly but can be observed with specialized detectors.
They are produced in great numbers in supernova explosions; about 99~\% of supernova energy is carried by neutrinos.

Supernova neutrinos are the only viable signal for geometric determination of the spacetime, given its conformal class.
(There are of course a number of different kinds of physical measurements that could achieve similar results. Most measurements rely on spectral information in one way or another. We require no information on spectrum or intensity.)
The resolution of current neutrino detection technology presents an issue, but the challenges are more technological than fundamental, unlike with any other particle.

We assume that supernova explosions are dense in the spacetime.
For a rough estimate on their true density, we may proceed as follows:
There are on the order of $0.01$ supernovae per year in our Milky Way and galaxies are typically about a million light years apart, which amounts to a density of
\begin{equation}
10^{-20}
\cdot
\frac{\text{supernovae}}{(\text{year})^4}
=
1
\cdot
\frac{\text{supernova}}{(10^5\text{ years})^4}
,
\end{equation}
where `year' stands for either year or light year ($c=1$).
On a scale of millions or billions of years and light years our model is not unreasonable.

A supernova explodes twice in our model:
The neutrinos are released first and photons a short time $\tau>0$ later.
In reality both processes are spread over time, but the time difference between the neutrino burst and the photon burst is substantial.
The difference of a couple of hours (depends on details of the model and the supernova) is due to photons being trapped in the expanding plasma in the early stages of the explosion, whereas neutrinos do not interact with the rest of the matter and are free to leave immediately.

We will treat the photon explosion as the main event and the neutrino explosion preceding it a small perturbation.
This choice makes the setting most compatible with using first pure photon observations to determine the conformal class of the spacetime.

There are several different types of supernovae, and we shall not venture into their taxonomy.
We will only mention that type Ia supernovae have very consistent characteristics due to the mechanism that produces them.
If we restrict our attention to this type only, it is reasonable to assume a constant and known time delay between the neutrino and the photon burst.

The simplifying assumptions of our model are the following:
\begin{itemize}
\item
The delay between neutrinos and photons is a known constant.
\item
All photons are released instantaneously, and so are all neutrinos.
\item
All neutrinos have the same known energy, and that energy is very high compared to the neutrino mass.
(The energies of the photons are geometrically irrelevant.)
\item
Supernovae are dense in the spacetime.
\end{itemize}

These assumptions make for a tractable geometric model, where the worldlines of neutrinos can be seen as small perturbations of worldlines of photons.
A photon travels along a null geodesic, so what we use to describe neutrino kinematics can well be called ultrarelativistic Jacobi fields.
We will explore the kinematics of supernova neutrinos within this model in the next subsection.

\subsection{Kinematics of a supernova explosion}

The kinematics will depend on two model parameters: time delay $\tau>0$ between neutrinos and photons and the mass-to-energy ratio $\mu=m/E$ of a neutrino in the supernova rest frame.
Both~$\tau$ and~$\mu$ are considered small: the delay is cosmologically negligible (a couple of hours certainly is) and the neutrinos are ultrarelativistic ($m<1\text{ eV}$ and $E\approx10\text{ MeV}$, so $\mu<10^{-10}$).
The speed of light need not be made explicit so we set $c=1$.

A supernova is massive enough to substantially curve the spacetime, so the spacetime near it is best described with the Schwarzschild or the Kerr metric.
Both of these metrics are asymptotically the Minkowski metric of the flat spacetime of special relativity.
However, on the largest scales the geometry is whatever the geometry of the whole universe is, and on that we impose no restrictions.

We do not focus on modelling the spacetime geometry and other physical aspects of the explosion itself, so the explosion is most conveniently described in the mesoscopic model of a Minkowski space between the very local (and highly curved and dynamic) and the global (and unrestricted) geometries.
Every tangent space of a Lorentzian manifold $(M,g)$ describing the spacetime is indeed a Minkowski space.

Consider a supernova explosion at $x\in M$.
Suppose the four-velocity --- the normalized velocity vector of the future-pointing worldline --- of the exploding star is $u\in T_sM$.
The star has non-zero mass, so~$u$ is always timelike and can be normalized.
We may always choose coordinates so that $u=(1,0,\dots,0)$, so that locally the coordinate time agrees with the proper time of the star.

For simplicity, let us first consider the case of $1+1$ dimensions in the Minkowski space~$T_xM$.
The worldline of a photon starting at the origin can be parametrized by $t\mapsto(t,t)$.
This parametrization is natural, as the parameter is the proper time of the exploding star --- or, equivalently, the coordinate time in its rest frame.

We want to find the worldline of a neutrino with energy~$E$ and mass~$m$ starting at $x=0$ when $t=-\tau$.
The Lorentz factor is $\gamma=E/m=\mu^{-1}$, so the neutrino's four-velocity is
\begin{equation}
u_\nu
=
\mu^{-1}
(
1
,
\sqrt{1-\mu^2}
)
\approx
\mu^{-1}
(
1
,
1-\tfrac12\mu^2
).
\end{equation}
The worldline can be parametrized with the proper time of the exploding star (rather than that of the neutrino itself):
\begin{equation}
t
\mapsto
\mu u_\nu t-(\tau,0)
=
(
t-\tau
,
(1-\tfrac12\mu^2)t
)
.
\end{equation}
Thus the spacetime separation between the photon and the corresponding neutrino is
\begin{equation}
\label{eq:Minkowski-kinematics}
(
-\tau
,
-\tfrac12\mu^2t
)
.
\end{equation}
The higher-dimensional Minkowskian description is similar.
Notice that~$\tau$ and~$\mu$ are the two small parameters of our model.

The photons are faster than neutrinos but are released later.
The time it takes fro the photons to reach the neutrinos is, roughly and in Minkowskian geometry, about~$2\tau\mu^{-2}$.
With $\mu<10^{-10}$ and~$\tau$ about an hour, this amounts to about $10^{16}$~years, which is about a million times the age of the universe.
Therefore photons do not catch up with neutrinos of such a high energy.

Globally the separation between the photon and the neutrino is described by a Jacobi field~$J(t)$.
If the photon is at~$\gamma(t)$, then the neutrino is at ``$\gamma(t)+J(t)$''.
(The sum only makes sense to leading order.)
Starting from the Minkowskian description~\eqref{eq:Minkowski-kinematics} of the kinematics and accounting for a general four-velocity~$u$ of the exploding star, we find that the initial conditions of the Jacobi field are:
\begin{equation}
\label{eq:nu-JF-initial-conditions}
\begin{cases}
J(t) = -\tau u
\\
D_tJ(t) = \tfrac12\mu^2\ip{u}{\dot\gamma(t)}(u-\dot\gamma(t)).
\end{cases}
\end{equation}
This \emph{neutrino Jacobi field} describes the kinematics of one neutrino in relation to one photon:
it is the infinitesimal deviation of the trajectory of a neutrino from the corresponding photon with the same initial direction.

In our measurements we do not know which photon corresponds to which neutrino.
Therefore we must consider the ``neutrino cone'' as a variation of the light cone.
Let us denote the future light cone of the supernova by
\begin{equation}
L_x^+
=
\{
v\in T_xM;
\ip{v}{v}=0,
\ip{v}{u}>0
\}.
\end{equation}
The \emph{neutrino variation field}~$V$ corresponding to $u\in T_xM$ and the fixed parameters~$\tau$ and~$\mu$ is the section of the pullback bundle $\exp_x|_{L_x}^{*}TM$ given by
\begin{equation}
V^x(v)
=
J(1),
\end{equation}
where~$J$ is the neutrino Jacobi field along the geodesic $t\mapsto\exp_x(tv)$.

Being a section of the pullback bundle means that the function $V\colon L_x\to TM$ satisfies $V(v)\in T_{\exp_x(v)}$.
If~$\exp_x$ is a local diffeomorphism near $v\in T_xM$, then the neutrino variation field~$V$ can be locally seen as the restriction of a vector field to the light cone~$\exp_x(L_x^+)$.
The heuristic interpretation is that a photon at~$\exp_x(v)$ (with $v\in L_x^+M$) corresponds to a neutrino at ``$\exp_x(v)+V^x(v)$''.
(Again, the sum is only sensible to leading order.)

The neutrino variation field describes the infinitesimal difference between the light cone and the neutrino cone.
The physically measurable data consists of a component of the neutrino variation field normal to the light cone~$\exp_x(L_x^+)$, and the data of two models are equivalent if these components agree.
The normal vectors and covectors of a light cone are lightlike, so there is no canonical choice of normalization, but the equality of the normal component is independent of the choice of a normal covector field.
We are now ready to define our neutrino data geometrically  and what it means for two models to have equivalent data.

\subsection{The data}

As described in the previous subsection, a covector field normal to a light cone is needed to restrict the measurement to the right component.
The light cone is $\exp_x(L_x)\subset M$ may fail to be a smooth submanifold where there are conjugate points along null geodesics, so such a covector field is best placed atop $L_s\subset T_xM$ instead of $\exp_x(L_x)\subset M$.

We say that $\nu\colon L_xM\to T^*M$ is a proper conormal field of the lightcone~$L_x$ if $\nu(v)=\rho(v)(\der\exp_x(v)v)^\flat$ for some non-vanishing smooth function $\rho\colon L_xM\to\R$.
It follows quickly from conformal invariance of null geodesics that this concept of a proper conormal field only depends on the conformal class of the Lorentzian metric.
Equivalence of neutrino data as defined below will be independent of the choice of this conormal field.
(It follows from a Lorentzian version of the Gauss lemma that~$\dot\gamma$ is normal to the light cone of any point on the null geodesic~$\gamma$. The proof is identical to the Riemannian version.)

\begin{definition}
\label{def:nu-data}
Fix any $\tau,\mu\in\R$.
Let $(M,g)$ be a smooth Lorentzian manifold, $U\subset M$ an open set, $x\in M$ a point, and $u\in T_sM$ a unit lightlike vector.
Denote $L_x^U=\exp_x^{-1}(U)\cap L_x^+$.
Let $\nu\colon L_x\to T^*M$ be a proper conormal field of the light cone~$L_xM$ and $V^x\colon L_x\to TM$ the neutrino variation field corresponding to $(u,\tau,\mu)$.

The neutrino data is the scalar map
\begin{equation}
\begin{split}
D(M,g,U,x,u,\nu,\tau,\mu)
\colon
&
L_x^U
\to
\R
,
\\
&
v
\mapsto
\nu(v)(V^x(v)).
\end{split}
\end{equation}
\end{definition}

\begin{definition}
\label{def:equivalent-data}
Fix any $\tau,\mu\in\R$.
Let $(M,g)$ be a smooth Lorentzian manifold, $U\subset M$ an open set, $x\in M$ a point, and $u\in T_sM$ a unit lightlike vector.
Denote $L_x^U=\exp_x^{-1}(U)\cap L_x^+$.
Let $\nu\colon L_x\to T^*M$ be a proper conormal field of the light cone~$L_xM$ and $V^x\colon L_x\to TM$ the neutrino variation field corresponding to $(u,\tau,\mu)$.

Let~$\hat g$ be a metric conformal to~$g$ and $\hat u\in T_xM$ a unit timelike vector so that $\ip{u}{\hat u}_g>0$ and so~$L_x^+$ is the same for both metrics.
Let $\beta\colon L_x\to L_x$ be a radial\footnote{A function $T_xM\to T_xM$ is radial when every vector is multiplied by a scalar depending on that vector.} function so that $\exp_x(v)=\widehat{\exp}_x(\beta(v))$ for all $v\in L_x^+$.
(This function exists and is unique by conformal invariance of null geodesics.)

We say that the $(M,g)$ and $(M,\hat g)$ have equivalent neutrino data if
\begin{equation}
D(M,g,U,x,u,\nu,\tau,\mu)
=
D(M,\hat g,U,x,\hat u,\nu\circ\beta,\tau,\mu).
\end{equation}
\end{definition}

\begin{remark}
Although the equivalence of data was formulated on the light cone $L_x\subset T_xM$ at a point~$x$ possibly far from the measurement set~$U$, the equivalence only concerns quantities measurable in~$U$.
If $\mathcal{L}=\exp_x(L_x)\cap U$ is a smooth submanifold, then equivalence may be formulated in terms of various fields on~$\mathcal{L}$.
A proper conormal field would simply be a non-vanishing covector field $\mathcal{L}\to T^*M$ so that $\ker(\nu(x))=T_x\mathcal{L}$.
The neutrino data can also be seen as a map $\mathcal{L}\to\R$, and equivalence of data means simply that the two functions are identical.
We chose to wrote the definitions above in terms of the tangent space~$T_xM$ rather than the measurement set~$U$ so as to allow for conjugate points and other non-smoothness.

In practice we assume that the conformal class has been determined by photon measurements, and so knowledge of the manifold~$M$, identification of the source point $x\in M$, and knowledge of the light cone~$L_x^+$ can be fairly assumed.
\end{remark}

\section{The uniqueness theorem}
\label{sec:thm}

We will show that if two metrics~$g$ and~$\hat g$ are conformal in a suitable set, then then equivalent neutrino data implies that they are equal.
The various sets should be thought of as follows:
\begin{itemize}
\item $M$ is the whole spacetime.
\item $U\subset M$ is where we measure.
\item $\Omega\subset M$ is where conformal equivalence of the two metrics is known.
\item $\omega\subset M$ is the union of all light rays along which we do neutrino measurements.
\end{itemize}

\begin{theorem}
\label{thm:main}
Fix any model parameters $\tau,\mu\in\R$.
Let~$M$ be a smooth manifold without boundary and $U\subset\Omega\subset M$ be open subsets.
Suppose~$\bar U$ is compact and contained in~$\Omega$.
Let~$g$ and~$\hat g$ be two Lorentzian metrics on~$M$.

The neutrino data determines the conformal factor uniquely in the following sense.
We make two assumptions:
\begin{enumerate}
\item
The two metrics are conformal:
\begin{itemize}
\item
$g=\hat g$ in~$U$
and
\item
$g=c\hat g$ in~$\Omega$ for some smooth function $c\colon\Omega\to(0,\infty)$.
\end{itemize}
Take any collection~$\Gamma$ of (not necessarily maximal) future-oriented lightlike geodesics in~$\Omega$ so that for all $\gamma\in\Gamma$ and~$t$ in the interval where~$\gamma$ is defined there is $t'>t$ so that $\gamma(t')\in U$.
Let $\omega\subset M$ be the union of all the rays in~$\Gamma$ and $\Omega\subset M$ any open subset containing~$\omega$.

\item
There are four-velocity fields so that the neutrino data is the same for both models:
Let~$u$ and~$\hat u$ be unit (w.r.t.~$g$ and~$\hat g$, respectively) timelike vector fields on $\omega\setminus U$.
These vector fields need not be even continuous.
Suppose that for all $x\in\omega\setminus U$ the neutrino data for
$(M,g,U,x,u,\nu,\tau,\mu)$
is equivalent with that of
$(M,\hat g,U,x,\hat u,\nu,\tau,\mu)$
in the sense of definition~\ref{def:equivalent-data}.
\end{enumerate}

Then two conclusions hold:
\begin{enumerate}
\item
The two metrics agree in~$\omega$:
The conformal factor satisfies $c|_{\omega}\equiv1$ and thus $g=\hat g$ in~$\omega$.

\item
The difference of four-velocities is normal to~$\Gamma$:
Take any $x\in\omega\setminus U$.
Whenever $\gamma(t)=x$ for some $\gamma\in\Gamma$, then $\ip{u(x)-\hat u(x)}{\dot\gamma(t)}=0$.
If the set of directions provided by~$\Gamma$ at~$x$ is an open subset of the light cone in~$T_xM$, then $u(x)=\hat u(x)$.
\end{enumerate}
\end{theorem}

We pose very little restrictions on the collection~$\Gamma$ of light rays.
It can well consist of only a single ray, in which case our result proves uniqueness along it.
If all light rays that meet~$U$ are included, then~$\Gamma$ provides an open subset of directions at all points and thus the stronger statement at the end of the theorem holds true.

\begin{proof}
Let~$\gamma$ be a $g$-geodesic and~$\hat\gamma$ a $\hat g$-geodesic, both lightlike.
We shift parameters so that $\gamma(0)=\hat\gamma(0)\in U$.
Much of the notation becomes lighter by using the abbreviation $\lambda(s)=\ip{u(s)}{\dot\gamma(s)}_g$.

Consider first just the geodesic~$\gamma(t)$ parametrized by an interval $I\subset\R$ and a supernova at $t=s$ along it.
The neutrino Jacobi field~$J_s$ along~$\gamma$ corresponding to this supernova has the initial conditions~\eqref{eq:nu-JF-initial-conditions}:
\begin{equation}
\begin{cases}
J_s(s) = -\tau u(s)
\\
D_tJ_s(s) = \tfrac12\mu^2\lambda(s)(u(s)-\dot\gamma(s)).
\end{cases}
\end{equation}

We define the auxiliary function $N\colon I^2\to\R$ by
\begin{equation}
N(t,s)
=
\ip{\dot\gamma(t)}{J_s(t)}_g.
\end{equation}
The Jacobi equation and antisymmetry of the Riemann curvature tensor give
\begin{equation}
\partial_t^2N(t,s)
=
0
\end{equation}
for all $t,s\in I$.
Therefore the initial conditions lead to the explicit expression
\begin{equation}
\label{eq:N-explicit}
N(t,s)
=
-\tau\lambda(s)
+
\tfrac12\mu^2\lambda(s)^2(t-s).
\end{equation}
This formula provides a crucial link between the data and the desired quantities.

Let us then turn to the geodesic $\hat\gamma(\hat t)$.
The parametrizations of the two geodesics differ by a diffeomorphism $\alpha\colon\hat I\to I$ so that $\hat\gamma(\hat t)=\gamma(\alpha(\hat t))$ and
\begin{equation}
\alpha'(\hat t)
=
c(\hat\gamma(t))
>
0
.
\end{equation}
(This change of parameters corresponds to the radial map~$\beta$ that matches the two parameters of any point on the light cone as seen in~$U$.)

We decorate all objects related to the metric~$\hat g$ by hats, including $\hat N(\hat t,\hat s)$ and $\hat\lambda(\hat s)$.
The inner products are all with respect to~$\hat g$ in these objects.
Equation~\eqref{eq:N-explicit} holds also when decorated with hats.

As the two metrics agree on~$U$,  we have $\alpha(\hat t)=\hat t$ when $t\approx0$.\footnote{Here and henceforth $t\approx0$ means that the stated identity holds for all~$t$ in a neighborhood of~$0$.}
As pointed out above, equivalence of data is independent of the choice of a proper conormal field.
We may thus set $\rho\equiv1$, so that $\nu=\dot\gamma^\flat$ when we only use equivalence along the line~$\gamma$.
This leads to the key identity
\begin{equation}
N(\hat t,\alpha(\hat s))
=
\hat N(\hat t,\hat s)
\end{equation}
whenever $\hat t\approx0$ and $\hat s\in\hat I$.
Equation~\eqref{eq:N-explicit} transforms this into
\begin{equation}
\label{eq:N=hatN}
-\tau\lambda(\alpha(\hat s))
+
\tfrac12\mu^2\lambda(\alpha(\hat s))^2(\hat t-\alpha(\hat s))
=
-\tau\hat\lambda(\hat s)
+
\tfrac12\mu^2\hat\lambda(\hat s)^2(\hat t-\hat s).
\end{equation}
Differentiating equation~\eqref{eq:N=hatN} with respect to~$\hat t$ --- which we may do as we have an open set of observation times $\hat t\approx0$ --- gives
\begin{equation}
\tfrac12\mu^2\lambda(\alpha(\hat s))^2
=
\tfrac12\mu^2\hat\lambda(\hat s)^2
\end{equation}
and thus
\begin{equation}
\label{eq:lamda=hatlambda}
\lambda(\alpha(\hat s))
=
\hat\lambda(\hat s).
\end{equation}
This simplifies equation~\eqref{eq:N=hatN} to
$
\hat t-\alpha(\hat s)
=
\hat t-\hat s
$,
whence~$\alpha$ is the identity function and thus $c=\alpha'=1$ along this geodesic.
Therefore $c\equiv1$ along all lines in~$\Gamma$ and thus on all of~$\omega$, concluding the proof of the first claim.

Now that the conformal factor is identically one, equation~\eqref{eq:lamda=hatlambda} yields
$\ip{u(\gamma(s))}{\dot\gamma(s)}=\ip{\hat u(\gamma(s))}{\dot\gamma(s)}$.
This holds for all curves in~$\Gamma$ through the same point $x\coloneqq\gamma(s)$, whence $u(x)-\hat u(x)$ is orthogonal to all~$\dot\gamma$ at~$T_xM$ for $\gamma\in\Gamma$.

The very last claim follows from the following lemma.
\end{proof}

Let us denote the Minkowski space of dimension $1+n$ by~$\R^{1,n}$ and the subset of null vectors in it by~$L$.

\begin{lemma}
Let $n\geq2$.
If $A\subset L$ is a non-empty open subset and $\ip{v}{a}=0$ for all $a\in A$, then $v=0$.
\end{lemma}

\begin{proof}
Take any interior point $a\in A\setminus\{0\}$.
Near~$a$ the light cone is a smooth hypersurface, and considering small variations of~$a$ shows that $\ip{v}{\xi}$ for all $\xi\in T_aL$.
The unique direction orthogonal to~$L$ at~$a$ is that of~$a$ itself, and so $v=\lambda a$ for some $\lambda\in\R$.

Now take any other interior point~$a'$ which is not a radial scaling of~$a$.
The same argument shows that~$v$ is a scalar multiple of~$a'$ as well.
Thus~$v$ lies on two different light rays through the origin and so $v=0$.
\end{proof}

\subsection*{Acknowledgements}

J.I.\ was supported by the Academy of Finland (grants 332890 and 351665).
G.U.\ was partially supported by NSF, a Walker Professorship at UW, a Si-Yuan Professorship at IAS, HKUST, and Simons Fellowship. 

\bibliographystyle{abbrv}
\bibliography{bibliography}

\end{document}